\documentclass[11pt,a4paper,reqno]{amsart}
\usepackage{amsthm,amsmath,amsfonts,amssymb,amsxtra,appendix,bookmark,euscript,delarray,dsfont,mathrsfs}
\usepackage{enumerate}
\usepackage{amssymb}
\usepackage[utf8]{inputenc}
\usepackage{xcolor}
\usepackage{graphicx}
\usepackage{upgreek}
\usepackage[normalem]{ulem}
\usepackage{ stmaryrd }

\usepackage{tikz}
\usetikzlibrary{decorations.pathreplacing, calc}

\theoremstyle{plain}
\newtheorem{theorem}{Theorem}
\newtheorem{lemma}{Lemma}
\newtheorem*{lemma*}{Lemma}

\newtheorem{proposition}{Proposition}[section]

\newtheorem*{conjecture*}{Conjecture}

\theoremstyle{definition}

\theoremstyle{remark}
\newtheorem{remark}{Remark}

\renewcommand{\d}{\mathrm{d}}                       

\newcommand{\dd}{\,\mathrm{d}}     

\DeclareMathOperator{\spec}{Spec}

\DeclareMathOperator{\supp}{supp}

\def\epsilon{\varepsilon}

\newcommand{\norm}[1]{\left\lVert #1 \right\rVert}
\newcommand{\abs}[1]{\left\lvert#1\right\rvert}
\renewcommand{\phi}{\varphi}
\newcommand\1{{\ensuremath {\mathds 1} }}

\newcommand{\ket}[1]{|#1 \rangle}       
\newcommand{\bra}[1]{\langle #1 |}      

\def\Neu{\mathrm{Neu}}

\def\C{\mathbb{C}}
\def\N{\mathbb{N}}

\def\R{\mathbb{R}}

\def\Z{\mathbb{Z}}

\def\a{\mathfrak{a}}

\def\cF {\mathcal{F}}
\def\cH {\mathcal{H}}

\def\cN {\mathcal{N}}

\def\eps {\varepsilon}

\newcommand		{\lt}			{\left}				
\newcommand		{\rt}			{\right}

\newcommand		{\bangle}[1]	{\lt\langle #1\rt\rangle}
\newcommand		{\weight}[1]	{\bangle{#1}}

\makeatletter 

\newcommand{\Rmnum}[1]{\expandafter\@slowromancap\romannumeral #1@}

\allowdisplaybreaks

\title[Kinetic Localization in Bose Gases]{Kinetic localization via Poincar\'e-type inequalities and applications to the condensation of Bose gases}

\author[J.J. Chong]{Jacky J. Chong}
\address{School of Mathematics and Statistics, Beijing Institute of Technology, Beijing, China}
\email{jwchong@bit.edu.cn}

\author[H. Liang]{Hao Liang}
\address{School of Mathematical Sciences, Peking University, Beijing, China}
\email{leunghao@stu.pku.edu.cn}

\author[P.T. Nam]{Phan Th\`anh Nam}
\address{Department of Mathematics, LMU Munich, Theresienstrasse 39, 80333 Munich, and Munich Center for Quantum Science and Technology (MCQST), Schellingstr. 4, 80799 Munich, Germany}
\email{nam@math.lmu.de}

\begin{document}

\begin{abstract} 
We propose a simplified localization method for Bose gases, based on a Poincare-type inequality, which leads to a new derivation of Bose--Einstein condensation for dilute Bose gases beyond the Gross--Pitaevskii scaling regime.
\end{abstract}

\subjclass[2020]{Primary 26D10; Secondary 46N50, 82B10, 81V70.}

\keywords{Poincare inequality, many-body quantum mechanics, dilute Bose gases, Bose--Einstein condensation,  Bogoliubov theory, Neumann boundary condition}

\maketitle
	
	
\section{Introduction}

The classical Poincar\'e inequality (also known as the Poincar\'e--Wirtinger inequality) asserts that for every bounded Lipchitz domain $\Omega\subset \R^d$, there exists a positive constant $C_{\Omega}>0$ such that 
\begin{align}\label{eq:Poincare}
\int_{\Omega} |u(x) - \left\langle u \right\rangle_{\Omega}|^2 \dd x  \le C_\Omega \int_{\Omega} |\nabla u(x)|^2 \dd x
\end{align}
for all $u\in H^1(\Omega)$ (see, e.g., \cite[Theorem 8.11]{lieb2001analysis}). Here, $\left\langle f \right\rangle_{A}:=|A|^{-1}\int_{A}f$  is the integral mean of $f$, with $|A|$ the Lebesgue measure of $A\subset \R^d$. For convenience, we write $\int_A f$ instead of $\int_A f\dd x$. Equivalently,  \eqref{eq:Poincare} can be interpreted as the existence of a spectral gap of size $C_\Omega^{-1}$ between the first eigenvalue and the second eigenvalue of the Neumann Laplacian  $-\Delta^{\Neu}_\Omega$ on $L^2(\Omega)$, namely
\begin{align}\label{eq:Poincare-Neumann}
 Q_\Omega  \le C_\Omega (-\Delta^{\Neu}_\Omega)\, , 
\end{align}
where $Q_A$ is an orthogonal projection on $L^2(\Omega)$ defined by
$$
Q_A := \1_A - P_A, \quad P_A := \frac{1}{|A|} \ket{\1_A}\!\bra{\1_A}
$$
 with $\1_A$ the indicator function of $A\subset \Omega$.   A general feature of \eqref{eq:Poincare-Neumann} is that the spectral gap $C_\Omega^{-1}$ decreases when the size of $\Omega$ increases. For example,  if $\Omega$ is a box of side length $L > 0$, then $C_\Omega^{-1} = \pi^2 L^{-2}$. 
 
 In the present paper, we are interested in a variant of \eqref{eq:Poincare-Neumann}, where the constant $C_\Omega$ is improved by incorporating additional information about $u$ on subdomains of $\Omega$. More precisely, assuming that the domain $\Omega$ is covered by a finite collection of disjoint subdomains $\{\Omega_j\}_{j \in J}$ with finite $J$, we aim to establish an operator inequality of the form
\begin{align}\label{eq:Poincare-Neumann-loc}
Q_\Omega \le \eps (- \Delta^{\Neu}_\Omega) + C_{\eps} \sum_{j\in J} Q_{\Omega_j} 
\end{align}
where the constant $\eps>0$ is significantly smaller than $C_\Omega$, with $C_{\eps}>0$ chosen suitably. 

The key idea behind \eqref{eq:Poincare-Neumann-loc} is that if a function is locally close to a constant, meaning $Q_{\Omega_j}$ is small for all $j$, then it must also be globally close to a constant, unless it has a very large kinetic energy (i.e. $-\Delta^{\Neu}_\Omega$ is large). While the underlying idea is transparent, constructing a quantitative version of \eqref{eq:Poincare-Neumann-loc} is nontrivial. For simplicity, we will restrict our analysis to the case where both $\Omega$ and its subdomains $\Omega_j$ are boxes in $\R^d$.

Our derivation of \eqref{eq:Poincare-Neumann-loc} is motivated by the study of Bose--Einstein condensation (BEC) in weakly interacting Bose gases. In this context, condensation is often derived using the spectral gap of the kinetic energy operator, which enables a comparison between the interacting system and a non-interacting one (possibly with a modified kinetic operator). Although proving BEC in the thermodynamic limit remains a major open problem in mathematical physics due to the absence of a spectral gap in the large-volume limit, it is nevertheless meaningful to study the problem in a finite-volume setting.

The significance of \eqref{eq:Poincare-Neumann-loc} lies in its ability to leverage condensation in smaller subdomains $\Omega_j$, where the proof is easier thanks to the existence of a stronger spectral gap, to deduce condensation in the larger domain $\Omega$ by a bootstrap argument. While our analysis does not cover the thermodynamic limit, it is sufficiently robust to treat simplified models that have attracted significant attention in recent decades, including the Gross--Pitaevskii scaling regime \cite{lieb2002proof, lieb2006derivation, nam2016ground, boccato2018complete, boccato2019bogoliubov, boccato2020optimal, nam2022optimal}, and slightly more singular regimes \cite{adhikari2021bose, fournais2020length, brennecke2024short,fournais2024lower}.  

The usefulness of a Poincaré inequality in the context of BEC is well known; see, for instance, the fundamental work \cite{lieb2003poincare}. However, the Poincaré inequality derived in our paper is simpler than existing ones, while remaining effective for analyzing the Bose gas. We hope this simplified approach will stimulate further research on localization arguments in BEC.

In the next section, we present a concrete formulation of \eqref{eq:Poincare-Neumann-loc} and provide a detailed discussion of its application to the derivation of BEC in interacting Bose gases. 

\subsection*{Acknowledgments} We would like to thank Zhenfu Wang for many helpful discussions. J. Chong thanks Søren Fournais for discussions on the localization methods used in his work, their connection to the IMS inequality, and their history. H. Liang thanks Zhibin Gong for insightful discussions. P. T. Nam thanks Christian Brennecke for insightful remarks comparing various proofs of Bose–Einstein condensation.    This work was partially supported by the National Key R\&D Program of China (Project No. 2024YFA1015500, J. Chong and H. Liang) and by the European Research Council through the ERC Consolidator Grant RAMBAS (Project No. 10104424, P.T. Nam).

\section{Main Results}

\subsection{Generalized Poincar\'e inequality for cubes} Our first result is a Poincare-type inequality for cubes and subcubes. To fix the notation, divide the unit box $\Lambda=[-1/2,1/2]^{d}\subset\R^{d}$ into $M^{d}$ close subcubes $\{\Lambda_i\}$ (with overlapping boundaries), each having side length $\ell=1/M$. 

\begin{theorem}[Poincare inequality for cubes] \label{thm:Poincare} 
    Let $d\in \mathbb{N}$ and $p>1$. There exist a constant $C_{p,d}>0$, depending only on the dimension $d$ and $p$, such that 
    \begin{equation}\label{ineq:Poincare}
        \lt\|f-\langle f\rangle_{\Lambda}\rt\|_{L^p(\Lambda)}^p
        \leq C_{p,d}  \norm{\nabla f }_{L^{p}(\Lambda)}\left( \frac{1}{\ell^{p}}\sum_{i=1}^{M^{d}}\norm{f-\left\langle f \right\rangle_{\Lambda_i} }_{L^{p}(\Lambda_{i})}^{p} \right)^{1-\frac{1}{p}}
	\end{equation}
	holds for all $f\in W^{1,p}(\Lambda)$ and $M\in\N$. Consequently,  for every $\eps>0$ we have the Poincar\'e inequality 
    \begin{equation}\label{eq:Poincare-Q-Qi}
        Q_\Lambda \leq \eps (-\Delta_\Lambda) + \frac{C^{2}_{2,d}}{4\eps \ell^{{2}}} \sum_{i=1}^{M^{d}}Q_{\Lambda_i}\,,
    \end{equation}
in the sense of quadratic forms on $L^2(\Lambda)$. Here, $\Delta_{\Lambda}$ is the Neumann Laplacian operator.
\end{theorem}

\begin{remark}
    The result also holds in the periodic setting. 
\end{remark}

\begin{remark}
    The bound \eqref{ineq:Poincare} is optimal for all dimension $d\ge 1$. For $M=2N$, we may define 
    \begin{equation}
        f_{N}(x):=\sum_{j=-N}^{N-1}\left[ \psi_{N}\left( x_{1}+\frac{j}{2N} \right)-\frac{j}{2N}\chi_{N}\left( x_{1}+\frac{j}{2N} \right)   \right], 
    \end{equation}
    where $\chi_{N}=\1_{[-1/2N,0]}$ is the characteristic function of interval $[-1/2N,0]$ and 
    \begin{equation*}
            \psi_{N}(x_{1})=2N\max\left\{x_{1}+\frac{1}{4N^{2}},0\right\}\chi_{N}(x_{1})\,.
    \end{equation*}
    The function $f_N(x)$ exhibits a staircase structure along the $x_1$-coordinate while remaining constant in all orthogonal directions. Constructed through successive translations of a basic building block $\psi_{N}$, the function alternates between horizontal plateaus and connecting segments with uniform slope $2N$. Each plateau spans an interval of length $1/2N$, while the connecting segments ensure continuity between adjacent levels. Inserting $f_{N}$, then we find that both side of \eqref{ineq:Poincare} have the same order of $N$. 
\end{remark}

\begin{remark}
    We note that the Poincar\'e inequality \eqref{ineq:Poincare} is closely related, albeit in completely different settings and form, to the multiscale Poincar\'e inequality, which is used in quantitative homogenization theory \cite{armstrong2016mesoscopic}. Similar estimates have also been developed to establish quantitative homogenization for interacting particle systems \cite{funaki2024quantitative}.
\end{remark}

\subsection{Bose--Einstein condensation of dilute Bose gases}  Let $\kappa\in[0,2/3)$ and fix $R \in (0, 1)$. Suppose $V$ is nonnegative, radially symmetric, and compactly supported on $\{ |x|<R\}$. We consider a system of $N$ bosons in $\Lambda:=[-1/2,1/2]^{3}$ interacting via a two-body potential of the form $N^{2-2\kappa}V(N^{1-\kappa}(x-y))$. The system is described by the Hamiltonian
\begin{equation}
	H_{N,\kappa}=\sum_{i=1}^{N}\left(-\Delta_{\Lambda}\right)_{x_{i}}+\sum_{1\leq i<j\leq N}N^{2-2\kappa}V(N^{1-\kappa}(x_{i}-x_{j}))
\end{equation}
acting on the bosonic space $L_{s}^{2}(\Lambda^{N})$, the Hilbert space consisting of functions in $L^{2}(\Lambda^{N})$ that are invariant with respect to all permutations of the $N$ particle labels $x_i \in \R^3$. Here, $-\Delta_{\Lambda}$ is the Laplacian with either the Neumann or periodic boundary condition. 

Note that the Hamiltonian $H_{N,\kappa}$ is unitary equivalent to $N^{2-2\kappa}$ times the Hamiltonian of $N$ bosons in a box $[-L/2,L/2]^3$ of side length $L=N^{1-\kappa}$, interacting through the unscaled potential $V$. In this sense, the Hamiltonian describes a system of very diluted Bose gas since, in the latter setting, the density of the system, $\rho:=N/L^3$, is equal to $N^{3\kappa-2}\ll 1$ for $\kappa\in [0,2/3)$. The case $\kappa=0$ corresponds to the Gross--Pitaevskii (GP) regime, and the case $\kappa=2/3$ corresponds to the usual thermodynamic limit. Hence, the parameter $\kappa$ allows us to interpolate between the GP regime and the thermodynamic regime. 

The study of low-energy properties, such as ground state energy $E_{N, \kappa}$ and its condensation phenomenon, of
dilute Bose gases described by the Hamiltonian $H_{N, \kappa}$ has a long history
in both the physics and mathematics literature. It can be shown rigorously that 
\begin{align}\label{eq:gs_energy}
    E_{N, \kappa} =4\pi\a_0 N^{\kappa+1} +4\pi\cdot\frac{128}{15\sqrt{\pi}}\,\a_0^{5/3}N^{5\kappa/2} +o\left(N^{5 \kappa/2}\right)_{N\rightarrow \infty}
\end{align}
for $\kappa>0$.
Here, $\a_0>0$ is the scattering length of $V$ defined by 
\begin{equation}
	8\pi \a_{0} :=\int_{\R^{3}}V(1-\omega)\, ,
\end{equation}
where $\omega$ is the unique radial solution to the zero-energy scattering equation  
\begin{equation}\label{eq:scattering-omega}
	-\Delta \omega=\tfrac{1}{2}V(1-\omega)\ {\rm in}\ \R^{3},\quad \lim\limits_{|x|\to\infty}\omega(x)=0\, 
\end{equation}
satisfying $0\le \omega \le 1$. The leading order term in the thermodynamic limit $\kappa =2/3$ was proved in \cite{dyson1957ground} for the upper bound and in \cite{lieb1998ground} for the lower bound. In the same setting, the next order correction term, also known as the Lee--Huang--Yang term \cite{lee1957eigenvalues}, was proved in \cite{yau2009second} (not including the hardcore interaction) for the upper bound and in \cite{fournais2020length, fournais2023energy} for the lower bound. Similar results were also recently obtained in the positive temperature setting for the free energy \cite{haberberger2023free,haberberger2024upper,fournais2024free}. Furthermore, in the GP case, \eqref{eq:gs_energy} must be adjusted by adding a finite-volume correction term that is of order one in the limit $N \to \infty$ (see \cite[Theorem 1.1]{boccato2019bogoliubov}).

We say that a state $\Psi \in L_{s}^{2}(\Lambda^{N})$ exhibits Bose--Einstein condensation (BEC) on the constant function $u_0= 1\in L^2(\Lambda)$ if  
\begin{equation}\label{eq:BEC}
	\left\langle \Psi, \cN_{+}\Psi \right\rangle=o(N)_{N\to \infty}\, , 
\end{equation}
where 
\begin{equation}
\cN_{+}:=\sum_{i=1}^{N}(Q_\Lambda)_{x_{i}}
\end{equation}
is the number operator associated with the particles excited out of the condensate.

\begin{theorem}[Bose--Einstein condensation]\label{thm:BEC} Assume $0\le V\in L^1(\R^3)$ is radially symmetric, decreasing, and with compact support. Let  $\kappa\in (0,2/11)$ and let $\Psi_{N}\in L^{2}_{s}(\Lambda^{N})$ be a normalized state such that 
	\begin{equation}
		\left\langle \Psi_{N}, H_{N,\kappa}\Psi_{N} \right\rangle\leq 4\pi\a_{0}N^{1+\kappa}+C_{0}N^{\frac{5\kappa}{2}+\frac{2-3\kappa}{4}}\,,
	\end{equation}
	with some constant $C_{0}>0$ independent of $N$. 	Then 
	\begin{equation}\label{eq:complete BEC}
	\lim_{N\to \infty}	\frac{\left\langle \Psi_{N}, \cN_{+}\Psi_{N} \right\rangle}{N} =0\,. 
	\end{equation}
\end{theorem}

The LHY formula suggests that, for the ground state, the “optimal” second-order term may scale as $N^{5\kappa/2}$. Here, however, we instead assume the condition $N^{5\kappa/2+(2-3\kappa)/4}$, which equals $\sqrt{N} \gg 1$ when $\kappa=0$. This specific constraint, $N^{5\kappa/2+(2-3\kappa)/4}$, arises from the error estimate in Proposition~\ref{prop:Neumann completing the square} with $S(\cdot)=\sqrt{\cdot}$. By treating the “completing-the-square” step more carefully (and taking into account the cubic terms), one may expect to improve Proposition~\ref{prop:Neumann completing the square} and choose $S(\cdot)=\log(\cdot)$, as already done in \cite{boccato2023bose}, using a more involved method based on unitary transformations. In that case, complete BEC can be obtained for $\kappa\in(0,\frac{2}{7}-\varepsilon)$, provided that the stricter upper bound $N^{5\kappa/2+\varepsilon}$ holds for the second-order term, with $\eps>0$ small arbitarily. That is, applying the result of \cite[Theorem~1.1]{boccato2023bose} in our proof of Theorem~\ref{thm:BEC} yields the following result.
\begin{theorem}\label{thm:improved BEC}
    Let $0\le V\in L^1\cap L^\infty(\R^3)$ be a radially symmetric function with compact support. Let  $\kappa\in (0,\frac27-\varepsilon)$, for any fixed $\varepsilon>0$, and let $\Psi_{N}\in L^{2}_{s}(\Lambda^{N})$ be a normalized state such that 
	\begin{equation}
		\left\langle \Psi_{N}, H_{N,\kappa}\Psi_{N} \right\rangle\leq 4\pi\a_{0}N^{1+\kappa}+C_{0}N^{\frac{5\kappa}{2}+\varepsilon}\,,
	\end{equation}
    with some constant $C_{0}>0$ independent of $N$. Then we obtain BEC \eqref{eq:complete BEC}. 
\end{theorem}

\subsection{Remarks on previous results} 
Despite the difficulty of proving BEC in the thermodynamic limit, significant progress has been made in recent years on regimes that interpolate between the GP and thermodynamic limits. The first proof of $100\%$ BEC in the ground state in the GP limit was obtained in \cite{lieb2002proof} (see also \cite{lieb1998ground, lieb2000bosons} for related results on the ground-state energy) and was later extended to approximate ground states for rotating Bose gases \cite{lieb2006derivation, nam2016ground}. More recently, the optimal rate of convergence \eqref{eq:BEC} was established in \cite{boccato2018complete, boccato2020optimal}, which was a key step in deriving the Bogoliubov excitation spectrum in the GP regime \cite{boccato2019bogoliubov}. Subsequently, several generalizations and simplified proofs in the GP regime were obtained in \cite{adhikari2021bose, michelangeli2019ground, caraci2021bose, hainzl2021another, brennecke2022bose, nam2022optimal, boccato2023bose, nam2023condensation, brennecke2024short}.

In the case of beyond-GP scaling, the proof of BEC—although not explicitly stated—was already implied by the work \cite{lieb2002proof} for $\kappa \in [0, \frac{1}{10})$ (see \cite[Chapter~7]{lieb2005mathematics}). A different approach, based on unitary renormalizations, was developed in \cite{adhikari2021bose} to prove BEC for approximate ground states in the regime $\kappa \in [0, \frac{1}{43})$; this was later revisited and improved in \cite{brennecke2024short} to cover the regime $\kappa \in [0, \frac{1}{20})$. Currently, the best available result concerning the range of $\kappa$ is due to Fournais in \cite{fournais2020length}, who established it for $\kappa \in [0, \frac{2}{5})$. Their work uses a Fock space argument to bypass the problem of distributing particles into small boxes, an approach inspired by \cite{lieb1998ground} and also essential for our proof given below.

Localization methods are typically helpful in proofs of Bose–Einstein condensation (BEC), although some works do not rely on them; see, for instance, \cite{adhikari2021bose, brennecke2024short} for alternative approaches that require neither $x$-localization nor the smallness of the interaction. (On the technical side, the approaches in \cite{boccato2018complete, nam2022optimal, boccato2023bose} also avoid localization in the actual proof of BEC, but these works rely on a smallness assumption on the interaction, which is incompatible with going beyond the Gross–Pitaevskii regime.)  The best-known result in \cite{fournais2020length}, in particular, was obtained using a subtle localization argument developed to derive the Lee--Huang--Yang formula \cite{brietzke2020second, brietzke2020simple, fournais2020energy, fournais2023energy}, building on even earlier works \cite{conlon1988n, lieb2001ground, lieb2004ground}. Historically, a more natural localization method based on Neumann box confinement was used in \cite{lieb1998ground}, which inspired the proof of BEC in \cite{lieb2002proof} and the generalization of the Poincaré inequality in \cite{lieb2003poincare}. More recently, new techniques concerning Neumann localization have been developed in \cite{haberberger2023free}, leading to an extension of the Lee--Huang--Yang formula to positive temperatures.

In the present paper, we combine the techniques from \cite{nam2022optimal, fournais2020length, haberberger2023free} with Inequality~\eqref{eq:Poincare-Q-Qi} in Theorem~\ref{thm:Poincare} to provide a simplified, self-contained proof of BEC for $\kappa\in(0,\frac{2}{11})$. More precisely, in Section~\ref{sect:Poincare}, we prove the Poincaré-type inequality in Theorem~\ref{thm:Poincare}, which is of independent interest. Then, in Section~\ref{sect:localization}, we explain how Theorem~\ref{thm:Poincare} can be used to implement a localization argument, allowing us to deduce BEC beyond the GP regime from a quantitative lower bound of the Hamiltonian imposed with Neumann boundary condition within the GP regime, namely when $\kappa=0$. To make our proof self-contained, we establish in Proposition~\ref{prop:Neumann completing the square} in Appendix~\ref{sect:GP} a quick proof of such lower bound when the scattering length of $V$ is sufficiently small, by combining techniques from \cite{nam2022optimal, fournais2020length, haberberger2023free}. This, together with the localization argument, implies BEC for $\kappa\in(0,\frac{2}{11})$. To improve this result (see Theorem~\ref{thm:improved BEC}), a refinement of Proposition~\ref{prop:Neumann completing the square} is needed. Theorem~\ref{thm:improved BEC} represents a minor modification of the proof of Theorem~\ref{thm:BEC} but relies heavily on \cite[Theorem~1.1]{boccato2023bose}.

More recently, \cite{junge2026propagation} extend our BEC results for $\kappa$ slightly larger than $2/5$ by incorporating the exact Lee--Huang--Yang term with an $O(\rho a (\rho a^3)^{1/2+\eta})$ remainder estimate via an overlapping Neumann localization.

\section{Poincaré-type Inequality}\label{sect:Poincare} 

In this section, we prove Theorem~\ref{thm:Poincare}.
\begin{proof}[Proof of Theorem~\ref{thm:Poincare}] 
Given $i\in\{1,2,\cdots,M^{d}\}$. Note that $\Lambda_{i}$ has at most $2d$ adjacent cubes.  We write $j\sim i$ if $\Lambda_{j}$ and $\Lambda_{i}$ have a shared face. Without loss of generality, take $\Lambda_i$ and $\Lambda_j$ such that 
    $$\Lambda_{i}\cup \Lambda_{j}=\left([a, b]\cup [b,c]\right)\times I_{2}\times\cdots\times I_{d}\,,$$
    where $b-a = c-b = \ell$ and some suitable intervals $I_{k} \subset \R$ such that $|I_k|=\ell$.

\noindent{\bf Step 1:} We first pass from the continuum to the local averages on boxes.   Let $f$ be a smooth function  such that $\nabla f \not\equiv 0$ on $\Lambda_i$ and define
    \begin{equation}
        g(t):=\ell^{-d+1}\int_{I_{2}\times\cdots\times I_{d}}f(t,t_{2},\cdots,t_{d})\dd t_{2}\cdots \d t_{d}
    \end{equation}
    on $I=[a, b]$. 
    By the Gagliardo--Nirenberg inequality and the fact that $u=g-\langle f\rangle_{\Lambda_i}$ has zero mean on $I$, we see that for all $p \in (1, \infty)$ there exists $C^{\rm GN}_{p}>0$ such that
    \begin{align}
        \left|g(b)-\langle f\rangle_{\Lambda_i}\right|^p\le \sup_{t \in I}\left|g(t)-\langle f\rangle_{\Lambda_i}\right|^p \le C^{\rm GN}_{p}\left(\int_I|g-\langle f\rangle_{\Lambda_i}|^p\right)^{1-\frac{1}{p}}\left(\int_I |g'|^p\right)^{\frac{1}{p}}.
    \end{align}
    Then applying the Cauchy--Schwarz inequality yields 
    \begin{align}\label{est:agmon-type}
        \left|g(b)-\langle f\rangle_{\Lambda_i}\right|^p\le \frac{C^{\rm GN}_{p}}{\ell^{d-1}}\left(\int_{\Lambda_{i}}\left|f-\left\langle f \right\rangle_{\Lambda_{i}} \right\vert^{p} \right)^{1-\frac{1}{p}}\left(\int_{\Lambda_i}|\nabla f|^{p}\right)^{\frac{1}{p}}\,.
    \end{align}
    The same argument applies to $\Lambda_j$. Summing \eqref{est:agmon-type} over  $\Lambda_i$ and $\Lambda_j$, which share a common side, gives the estimate  
    \begin{align*}
                \left|\langle f\rangle_{\Lambda_j}-\langle f\rangle_{\Lambda_i}\right|^p \le \frac{C_{p,d}}{\ell^{d-1}}\left(\int_{\Lambda_{i}}\left|f-\left\langle f \right\rangle_{\Lambda_{i}} \right\vert^{p}+\int_{\Lambda_{j}}\left|f-\left\langle f \right\rangle_{\Lambda_{j}} \right\vert^{p} \right)^{1-\frac{1}{p}}\left(\int_{\bigcup_{k\sim i}\Lambda_k}|\nabla f|^{p}\right)^{\frac{1}{p}}\,.
    \end{align*}

    Summing over $j$ such that $j\sim i$, we have
    \begin{equation}\label{est:ratio}
        \frac{\left(\sum_{j:j\sim i}\left|\left\langle f\right\rangle_{\Lambda_{i}}-\left\langle f\right\rangle_{\Lambda_{j}}\right\vert\right)^{p}}{\left(\int_{\bigcup_{k\sim i}\Lambda_k}|\nabla f|^{p}\right)^{\frac{1}{p}}}\leq\frac{C_{p,d}}{\ell^{d-1}}\sum_{j:j\sim i}\left(\int_{\Lambda_{j}}\abs{f-\left\langle f\right\rangle_{\Lambda_{j}}}^{p}\right)^{1-\frac{1}{p}}\,.
    \end{equation}
    Summing over $i$, it follows by H\"older's inequality and \eqref{est:ratio} we have the bound \begin{equation}\label{est:l2_bound_on_discrete_Laplacian}
    \begin{aligned}
    \sum\limits_{i=1}^{M^{d}}\left(\sum_{j:j\sim i}\left|\left\langle f \right\rangle_{\Lambda_{i}}-\left\langle f \right\rangle_{\Lambda_{j}}  \right\vert\right)^{p}  \le 
        \frac{C_{p,d}}{\ell^{d-p}}\left(\frac{1}{\ell^{p}}\sum_{i=1}^{M^{d}}\int_{\Lambda_{i}}\left|f-\left\langle f \right\rangle_{\Lambda_{i}} \right\vert^{p}\right)^{1-\frac{1}{p}}\left(
        \int_{\Lambda}\abs{\nabla f}^{p}\right)^{\frac{1}{p}}.
    \end{aligned}
    \end{equation}
\bigskip

\noindent{\bf Step 2:} The previous step reduced the problem on $\Lambda$ to a problem on the grid $\llbracket M\rrbracket^d\subset \Z^d$. 
To obtain estimates on $\llbracket M\rrbracket^d$, we use the fact that the Cheeger constant of $\llbracket M\rrbracket^d$ is inversely proportional to $M$, namely we have the following discrete Poincar\'e-type inequality on $\llbracket M\rrbracket^d$. For all $d\ge 1$ and $p>1$, there exists $C_{p, d}>0$ such that the following holds
    \begin{equation}\label{est:discrete_poincare}
        \left(\sum_{i=1}^{M^{d}}\left|\left\langle f \right\rangle_{\Lambda}-\left\langle f \right\rangle_{\Lambda_{i}}  \right\vert^{p}\right)^{\frac{1}{p}}\le C_{p,d} M \left(\sum_{i=1}^{M^{d}}\left(\sum_{j\sim i}\left|\left\langle f \right\rangle_{\Lambda_{i}}-\left\langle f \right\rangle_{\Lambda_{j}}  \right\vert\right)^{p}\right)^{\frac{1}{p}}\,. 
    \end{equation}
A direct proof of \eqref{est:discrete_poincare} can be founded in Appendix \ref{appendix:discretepoincare}. Combining \eqref{est:l2_bound_on_discrete_Laplacian}, \eqref{est:discrete_poincare}, and the fact $M=1/\ell$ yields
\begin{align}\label{est:difference_local_and_global_averages-(p,p)}
    \left(\frac{1}{M^d}\sum_{i=1}^{M^{d}}\left|\left\langle f \right\rangle_{\Lambda_{i}}-\left\langle f \right\rangle_{\Lambda}  \right\vert^{p}\right)^{\frac{1}{p}} \le C_{p, d}\left(\frac{1}{\ell^{p}}\sum_{i=1}^{M^{d}}\int_{\Lambda_{i}}\left|f-\left\langle f \right\rangle_{\Lambda_{i}} \right\vert^{p}\right)^{\frac{1-\theta}{p}}\|\nabla f\|_{L^{p}(\Lambda)}^{\theta}\,,
\end{align}
where $\theta = \frac{1}{p} \in (0, 1)$. Using the elementary inequality $(a+b)^p\leq 2^{p-1}(a^p+b^p)$ for non-negative $a$ and $b$, together with \eqref{est:difference_local_and_global_averages-(p,p)}, we have 
\begin{align*}
    &\norm{f-\weight{f}_{\Lambda}}_{L^p(\Lambda)}^p\\
    &\leq 2^{p-1}\left(\sum_{i=1}^{M^d}\int_{\Lambda_i}\left|f-\weight{f}_{\Lambda_i}\right\vert^p+\frac{1}{M^d}\sum_{i=1}^{M^d}\left|\weight{f}_{\Lambda}-\weight{f}_{\Lambda_i}\right\vert^p\right)\\
    &\leq 2^{p-1}\left(\sum_{i=1}^{M^d}\int_{\Lambda_i}\left|f-\weight{f}_{\Lambda_i}\right\vert^p+C_{p,d}\norm{\nabla f}_{L^p(\Lambda)}\left(\frac{1}{\ell^p}\sum_{i=1}^{M^d}\int_{\Lambda_i}\left|f-\weight{f}_{\Lambda_i}\right\vert^p\right)^{1-\frac{1}{p}}\right)\\
    &=2^{p-1}\left(C_{p,d}\norm{\nabla f}_{L^p(\Lambda)}+\ell^{p-1}\left(\sum_{i=1}^{M^d}\int_{\Lambda_i}\left|f-\weight{f}_{\lambda_i}\right\vert^p\right)^{\frac{1}{p}}\right)\left(\frac{1}{\ell^p}\sum_{i=1}^{M^d}\int_{\Lambda_i}\left|f-\weight{f}_{\Lambda_i}\right\vert^p\right)^{1-\frac{1}{p}}\,.
\end{align*}
Recall that $\ell\leq 1$. The classical Poincar\'e inequality implies
\begin{align*}
    \ell^{p-1}\left(\sum_{i=1}^{M^d}\int_{\Lambda_i}\left|f-\weight{f}_{\lambda_i}\right\vert^p\right)^{\frac{1}{p}}\leq \ell^{p-1}\left( C_{p,d}^p\ell^p\sum_{i=1}^{M^d}\int_{\Lambda_i}\left|\nabla f\right\vert^p\right)^{\frac{1}{p}}\leq C_{p,d}\norm{\nabla f}_{L^p(\Lambda)}.
\end{align*}
We obtain \eqref{ineq:Poincare} for all smooth functions $f$ with $\nabla f\not\equiv 0$. The argument extends to all functions $f\in W^{1,p}(\Lambda)$, which concludes the proof of \eqref{ineq:Poincare}.

\noindent{\bf Step 3:} Choosing $p=2$ in \eqref{ineq:Poincare} and using the Cauchy--Schwarz inequality, we have 
    \begin{align*}
        \int_{\Lambda}\left|f-\left\langle f \right\rangle_{\Lambda} \right\vert^{2}&\leq C_{2,d}\ell^{-1}\sqrt{\left( \int_{\Lambda}|\nabla f|^{2} \right)\left( \sum_{i=1}^{M^{d}}\int_{\Lambda_i}\abs{f-\left\langle f \right\rangle_{\Lambda_i} }^{2} \right)  }\\
        &\leq \eps \int_{\Lambda}|\nabla f|^{2} + \frac{C_{2,d}^{2}}{ 4\eps \ell^{2}} \sum_{i=1}^{M^{d}}\int_{\Lambda_i}\abs{f-\left\langle f \right\rangle_{\Lambda_i} }^{2}
    \end{align*}
    for all $f\in H^{1}(\Lambda)$ and $\eps>0$. This is equivalent to \eqref{eq:Poincare-Q-Qi}. 
 \end{proof}

\section{Proof of Theorem \ref{thm:BEC}}\label{sect:localization}

We use the Neumann box localization argument mentioned in the previous section to prove Theorem~\ref{thm:BEC}, that is, we cut the large box of side length 1 into several small cells of side length $\ell$. In each small cell, the localized Hamiltonian (rescaled) describes the system in the GP regime for some suitable choice of $\ell$ and the kinetic energy operator is described by the Neumann Laplacian on the small cell.

On the technical side, the localization of the interaction potential energy is simple due to the nonnegativity of $V$. The main difficulty is to localize the number operator $\cN_{+}$, for which we will sacrifice a little bit of kinetic energy and use the Poincare-type inequality in Theorem \ref{thm:Poincare}.

\subsection{Kinetic energy localization}

From now on, we only consider $d=3$. 

Let $\Lambda=\bigcup _{i=1}^{M^{3}}\Lambda_i$. As a multiplication operator, $\1_{\Lambda_{i}}$ defines an orthogonal projection $\1_{\Lambda_{i}}:L^{2}(\Lambda)\to L^{2}(\Lambda_i)$. Recall $P_{\Lambda_{i}},Q_{\Lambda_{i}}$ are the projections $P_{\Lambda_{i}},Q_{\Lambda_{i}}:L^{2}(\Lambda)\to L^{2}(\Lambda_i)$ defined by 
\begin{equation}
    Q_{\Lambda_{i}} = \1_{\Lambda_{i}} - P_{\Lambda_{i}}, \quad P_{\Lambda_{i}} = \frac{1}{|\Lambda_{i}|} \ket{\1_{\Lambda_{i}}}\!\bra{\1_{\Lambda_{i}}}\, . 
\end{equation}
It is easy to see that $\{\1_{\Lambda_{i}}\}$ induce a unitary isomorphism:
\begin{equation}\label{unitaryiso}
    L^{2}(\Lambda)\simeq \bigoplus \limits_{i=1}^{M^{3}}\1_{\Lambda_{i}}\left(L^{2}(\Lambda)\right)\simeq \bigoplus_{i=1}^{M^{3}}L^{2}(\Lambda_i). 
\end{equation}

We have the following useful operator lower bound, which is a direct result of our Poincare-type inequality \eqref{ineq:Poincare}.

\begin{lemma}[Kinetic energy localization]\label{kineticloc}
    There exists a constant $C$, dependent on the dimesion $d$, such that for any $\alpha\geq0,\ell<1/2$, the following lower bound (in the sense of quadratic forms) holds:
    \begin{equation}
        -\Delta_{\Lambda}-\ell^{2+\alpha}Q_{\Lambda}\geq \left( 1-C\ell^{4+2\alpha} \right)\sum_{i=1}^{M^{3}}\left( -\Delta_{\Lambda_{i}}^{\Neu}-\frac{\pi/8}{\ell^{2}}Q_{\Lambda_{i}} \right),  
    \end{equation}
    where $\ell=1/M$ and we used \eqref{unitaryiso} to identify $L^{2}(\Lambda)$ and $\bigoplus_{i=1}^{M^{3}}L^{2}(\Lambda_i)$.
\end{lemma}

\begin{proof}
    As in Step 3 of the proof of Theorem~\ref{thm:Poincare}, we have 
    \begin{align*}
        \int_{\Lambda}\left|f-\left\langle f \right\rangle_{\Lambda} \right\vert^{2}
        &\leq\frac{1}{\ell^{2+\alpha}}\frac{\pi/4}{\ell^{2}}\sum_{i=1}^{M^{3}}\int_{\Lambda_{i}}\abs{f-\left\langle f \right\rangle_{\Lambda_{i}} }^{2}+C\ell^{2+\alpha}\int_{\Lambda}|\nabla f|^{2}\,,
    \end{align*}
    for all $f\in H^{1}(\Lambda)$, which is equivalent to the following operator inequality:
    \begin{equation}
        Q_{\Lambda}\leq\frac{1}{\ell^{2+\alpha}}\frac{\pi/4}{\ell^{2}}\sum_{i=1}^{M^{3}}Q_{\Lambda_{i}}+C\ell^{2+\alpha}\left(-\Delta_{\Lambda}\right)\,.
    \end{equation}
    Then, we obtain 
    \begin{align*}
        -\Delta_{\Lambda}-\ell^{2+\alpha}Q_{\Lambda}&\geq \left( 1-C\ell^{4+2\alpha} \right)\sum_{i=1}^{M^{3}}-\Delta_{\Lambda_{i}}^{\Neu}-\frac{\pi/4}{\ell^{2}}\sum_{i=1}^{M^{3}}Q_{\Lambda_{i}}\\
        &= \left( 1-C\ell^{4+2\alpha} \right)\sum_{i=1}^{M^{3}}\left(-\Delta_{\Lambda_{i}}^{\Neu}-\frac{1}{1-C\ell^{4+\alpha}}\frac{\pi/4}{\ell^{2}}Q_{\Lambda_{i}}\right)\\
        &\geq\left( 1-C\ell^{4+2\alpha} \right)\sum_{i=1}^{M^{3}}\left(-\Delta_{\Lambda_{i}}^{\Neu}-\frac{\pi/8}{\ell^{2}}Q_{\Lambda_{i}}\right)\, ,
    \end{align*}
    which completes the proof of the lemma.
\end{proof}

 \begin{remark}
   This result should be compared with \cite[Lemma 3.3]{fournais2020length} when restricted to the periodic setting. Since our proof of the kinetic localization lemma relies on the Poincaré-type inequality \eqref{ineq:Poincare}, we must necessarily expend a small portion of the kinetic energy to control $Q_{\Lambda}$. The original goal, as in \cite{fournais2020length}, was to control the sub-Hamiltonians via the gap $-\Delta_{\Lambda}-Q_{\Lambda}$; however, because the Poincaré-type inequality consumes part of the kinetic energy, we can only control the sub-Hamiltonians using $-\Delta_{\Lambda}-\ell^{2+\alpha} Q_{\Lambda}$, which naturally yields a weaker bound as $\ell \to 0$. This is because we now face a trade-off: $\ell$ must be small so that $\ell^{2+\alpha}$ multiplied by the leading order energy remains subleading relative to the Lee--Huang--Yang term, but also large enough to ensure the existence of BEC.
 \end{remark}

\subsection{Proof of Theorem~\ref{thm:BEC}}

Let $V_{N,\kappa}(x):=N^{2-2\kappa}V(N^{1-\kappa}x)$, whose scattering length $\a$ satisfies $\a=\a_{0}/N^{1-\kappa}$. Given a parameter $\rho_{\mu}>0$, consider the following Hamiltonian operator $\cH_{\rho_{\mu}}$ define on the symmetric Fock space $\cF_{s}(L^{2}(\Lambda)):=\C\oplus\bigoplus^\infty_{n=1}L^2_s(\Lambda^n)$. The operator $\cH_{\rho_{\mu}}$ commutes with the particle number and satisfies
\begin{equation}
    (\cH_{\rho_{\mu}})_{n}:=\sum_{i=1}^{n}(\left(-\Delta_{\Lambda}\right)_{x_{i}}-\ell^{2+\alpha}(Q_{\Lambda})_{x_{i}})+\sum_{1\leq i<j\leq n}V_{N,\kappa}(x_{i}-x_{j})-8\pi\a\rho_{\mu}n.
\end{equation}
Recall from \eqref{unitaryiso} that we have $L^{2}(\Lambda)\cong \bigoplus_{i=1}^{M^{3}}L^{2}(\Lambda_i)$, which implies the following isomorphism of Fock spaces:
\begin{equation}
    \cF_{s}(L^2(\Lambda))\cong \bigotimes \limits_{k=1}^{M^{3}}\cF_{s}(L^2(\Lambda_i))\,.
\end{equation}
Using this identification and Lemma~\ref{kineticloc}, we have
\begin{align*}
  (\cH_{\rho_{\mu}})_{n}\geq&\, \sum_{i=1}^{M^{3}}\left[ \left( 1-C\ell^{4+2\alpha} \right)\sum_{j=1}^{n}\left( (-\Delta_{\Lambda_{i}}^{\Neu})_{x_{j}}-\frac{\pi/8}{\ell^{2}}(Q_{\Lambda_{i}})_{x_{j}}-8\pi\a\rho_{\mu} \right) \right]\\
    &\, +\sum_{i=1}^{M^{3}}\left[\sum_{1\leq j<l\leq n}\1_{\Lambda_{i}}(x_{j})V_{N,\kappa}(x_{j}-x_{l})\1_{\Lambda_{i}}(x_{l}) \right]-C\ell^{4+2\alpha}\a\rho_{\mu}n\\
    =:&\, \left( 1-C\ell^{4+2\alpha} \right)\sum_{i=1}^{M^{3}}\left( \cH_{\rho_{\mu}}(\Lambda_i) \right)_{n}-C\ell^{4+2\alpha}\a\rho_{\mu}n\,,
\end{align*}
which could be lifted to the Fock space as
\begin{equation}\label{Hamiltonloc}
    \cH_{\rho_{\mu}}\geq \left( 1-C\ell^{4+2\alpha} \right)\sum_{i=1}^{M^{3}}\cH_{\rho_{\mu}}(\Lambda_i)-C\ell^{4+2\alpha}\a\rho_{\mu}{\cN}. 
\end{equation}
Here, the $n$-body sector of $\cH_{\rho_{\mu}}(\Lambda_{i})$ is defined by
\begin{multline}
        \left( \cH_{\rho_{\mu}}(\Lambda_i) \right)_{n}:=\sum_{j=1}^{n}\left( (-\Delta_{\Lambda_{i}}^{\Neu})_{x_{j}}-\frac{\pi/8}{\ell^{2}}(Q_{\Lambda_{i}})_{x_{j}} \right)\\+\sum_{1\leq j<l\leq n}\left(\1_{\Lambda_{i}}V_{N,\kappa}\1_{\Lambda_{i}}\right)_{x_{j},x_{l}}-8\pi\a\rho_{\mu}n\,,
\end{multline}
with the two-body multiplication operator $\left(\1_{\Lambda_{i}}V_{N,\kappa}\1_{\Lambda_{i}}\right)_{x_{j},x_{l}}=\1_{\Lambda_{i}}(x_{j})V_{N,\kappa}(x_{j}-x_{l})\1_{\Lambda_{i}}(x_{l})$.

\begin{lemma}\label{lemmasmallenergy}
    Let $\ell=K^{-1}\frac{1}{\sqrt{\rho_{\mu}\a}}$ with $K$ being large ($K=20$ suffices). Then
    \begin{equation}
        \cH_{\rho_{\mu}}(\Lambda_{i})\geq-4\pi\rho_{\mu}^{2}\a\ell^{3}-C\rho_{\mu}^{2}\a\ell^{3}\left( \rho_{\mu}\a^{3} \right)^{\frac{1}{2}}S \left( \left( \rho_{\mu}\a^{3} \right)^{-1/2}  \right)  
    \end{equation}
    for a sufficiently small value of $\rho_{\mu}\a^{3}$. Here, $S(\cdot)$ is the same as in Proposition~\ref{prop:Neumann completing the square}  and the result is uniform in $i\in\{1,\cdots,M^{3}\}$.
\end{lemma}

\begin{proof}
    We follow a similar approach as in the proof given in \cite[Theorem 2.1]{fournais2020length}.
    Since all these operators are unitary equivalent, it suffices to consider $i=1$.  Notice that $\cH_{\rho_{\mu}}(\Lambda_{1})$ commutes with the particle number operator, so it suffices to establish the lower bound for $\left\langle \Psi, (\cH_{\rho_{\mu}}(\Lambda_{1}))_{N}\Psi \right\rangle $ for $\Psi\in L^{2}_{s}(\Lambda_{1}^{N})$ for some arbitrary particle number $N$. Since $\rho_{\mu}\ell^{3}=K^{-3}(\rho_{\mu}\a^{3})^{-1/2}\gg 1$, we partition the $N$ particles into groups with  $\rho_{\mu}\ell^{3}$ order number of particles. Let 
    \begin{equation*}
        \{1,\cdots,N\}=\bigcup \limits_{j=1}^{\xi}S_{j},\ \ \ S_{j}\cap S_{k}=\varnothing \ {\rm for}\ j\neq k
    \end{equation*}
    and 
    \begin{equation*}
        |S_{\xi}|\leq 4\rho_{\mu}\ell^{3},\ \ |S_{j}|\in \left[ 3\rho_{\mu}\ell^{3},4\rho_{\mu}\ell^{3} \right] \ {\rm for}\ j<\xi.
    \end{equation*}
    Using the positivity of $V$, we obtain the lower bound
    \begin{equation}\label{smallenergy1}
 \left\langle \Psi,(\cH_{\rho_{\mu}}(\Lambda_{1}))_{N}\Psi \right\rangle\geq\sum_{j=1}^{\xi}{\rm inf}\spec (\cH_{\rho_{\mu}}(\Lambda_{1}))_{|S_{j}|}\,
    \end{equation}
    which is simply a consequence of convexity and of dropping interactions
among different small boxes of particles.   Hence, it suffices to reduce the analysis to $(\cH_{\rho_{\mu}}(\Lambda_{1}))_{n}$ with $n\leq 4\rho_{\mu}\ell^{3}$, i.e. $\frac{n\a}{\ell}\leq \frac{1}{100}.$

    After rescaling, $(\cH_{\rho_{\mu}}(\Lambda_{1}))_{n}+8\pi\a\rho_{\mu}n$ is unitary equivalent to $\frac{1}{\ell^{2}}H_{n,\ell}^{\Neu}$, where 
    \begin{equation}
        H_{n,\ell}^{\Neu}:=\sum_{i=1}^{n}\left(-\Delta_{x_{i}}^{\Neu}-\frac{\pi}{8}Q_{x_{i}}\right)+\sum_{1\leq i<j\leq n}\left( \ell N^{1-\kappa} \right)^{2}V \left( \ell N^{1-\kappa}(x_{i}-x_{j}) \right)  
    \end{equation}
    acts on $L^{2}_{s}(\Lambda^{n})$. Notice that $\frac{\ell}{\a}=\frac{1}{20}\frac{1}{\sqrt{\rho_{\mu}\a^{3}}}\gg 1$ and $\frac{n \a}{\ell}\leq\frac{1}{100}$, we have 
    \begin{equation}
        n\leq {\rm min}\left\{\frac{1}{48}\frac{\ell}{\a},\frac{\pi}{8C}\frac{\ell}{\a}\log\left(\frac{\ell}{\a}\right)\right\}
    \end{equation}
    for sufficiently small value of $\rho_{\mu}\a^{3}$. Here $C$ is the constant in Proposition~\ref{prop:Neumann completing the square} . Choosing $\mu=\frac{\pi+8\pi\a n/\ell}{2}$ in Proposition~\ref{prop:Neumann completing the square} , we have 
    \begin{align*}
        H_{n,\ell}^{\Neu}&\geq 4\pi \a\frac{n^{2}}{\ell}+ \left( \frac{\pi}{4}-\frac{\pi}{8}-\frac{Cn}{\frac{\ell}{\a}\log \left( \frac{\ell}{\a} \right) } \right)\cN_{+}-C \left( \frac{\a n}{\ell} \right)^{2}S \left( \frac{\ell}{\a} \right)-C\frac{\a n}{\ell}\\
        &\geq4 \pi\a \frac{n^{2}}{\ell}-C \left( S \left( \frac{\ell}{\a}+1 \right)  \right).    
    \end{align*}
    Thus 
    \begin{equation}
        \begin{aligned}
        \left( \cH_{\rho_{\mu}}(\Lambda_{1}) \right)_{n}&\geq4\pi\a\frac{n^{2}}{\ell^{3}}-\frac{C}{\ell^{2}} \left( S \left( \frac{\ell}{\a} \right)+1  \right) -8\pi\a\rho_{\mu} n\\
        &\geq2\pi\frac{\a}{\ell^{3}}\left( \rho_{\mu}\ell^{3}-n \right)^{2}-4\pi\a\rho_{\mu}^{2}\ell^{3}-\frac{C}{\ell^{2}}\left( S \left( \frac{\ell}{\a} \right)+1  \right).   
        \end{aligned}
    \end{equation}
    If $3\rho_{\mu}\ell^{3}\leq n\leq 4\rho_{\mu}\ell^{3}$, we have 
    \begin{equation}\label{smallenergy2}
        \begin{aligned}
        \left( \cH_{\rho_{\mu}}(\Lambda_{1}) \right)_{n}&\geq2\pi\frac{\a}{\ell^{3}}\left( \rho_{\mu}\ell^{3}-n \right)^{2}-4\pi\a\rho_{\mu}^{2}\ell^{3}-\frac{C}{\ell^{2}}\left( S \left( \frac{\ell}{\a} \right)+1  \right)\\
        &\geq4\pi\a\rho_{\mu}^{2}\ell^{3}-\frac{C}{\ell^{2}}\left( S \left( \frac{\ell}{\a} \right)+1  \right)\\
        &=\rho_{\mu}\a \left[ 4\pi K^{-3}\left( \rho_{\mu}\a^{3} \right)^{-1/2}-K^{2}\left( S \left( \frac{1}{K\sqrt{\rho_{\mu}\a^{3}}} \right)+1  \right) \right]\geq 0      
        \end{aligned}
    \end{equation}
    for sufficiently small value of $\rho_{\mu}\a^{3}$.
    
    If $n\leq 4\rho_{\mu}\ell^{3}$, we have 
    \begin{equation}\label{smallenergy3}
        \begin{aligned}
            \left( \cH_{\rho_{\mu}}(\Lambda_{1}) \right)_{n}&\geq-4\pi\a\rho_{\mu}^{2}\ell^{3}-\frac{C}{\ell^{2}}\left( S \left( \frac{\ell}{\a} \right)+1  \right)\\
            &= -4\pi\rho_{\mu}^{2}\a\ell^{3}-C\rho_{\mu}^{2}\a\ell^{3}\left( \rho_{\mu}\a^{3} \right)^{\frac{1}{2}}S \left( \left( \rho_{\mu}\a^{3} \right)^{-1/2}  \right).
        \end{aligned}
    \end{equation}
    From \eqref{smallenergy1}, \eqref{smallenergy2} and \eqref{smallenergy3} we obtain the desired result.
\end{proof}

Now we can conclude the proof of Theorem \ref{thm:BEC}. 
\begin{proof}[Proof of Theorem \ref{thm:BEC}]
    It follows by \eqref{Hamiltonloc} and Lemma \ref{lemmasmallenergy} that 
    \begin{equation}
        \cH_{\rho_{\mu}}\geq \left( 1-C\ell^{4+2\alpha} \right)\left[ -4\pi\a\rho_{\mu}^{2}-C\rho_{\mu}^{2}\a \left( \rho_{\mu}\a^{3} \right)^{1/2}S \left( \frac{1}{\sqrt{\rho_{\mu}\a^{3}}} \right)   \right]  -C\ell^{4+2\alpha}\a\rho_{\mu}\cN.
    \end{equation}
    Here we used the fact that $M^{3}=1/\ell^{3}$. Restricting to the $N$ body section and choosing $\rho_{\mu}=N$ ($\rho_{\mu}\a^{3}\sim N^{3\kappa-2}$ is small, provided $N$ is large enough), we have
    \begin{equation}
        H_{N,\kappa}-\ell^{2+\alpha}\cN_{+}\geq 4\pi\a_{0} N^{1+\kappa}-C\ell^{4+2\alpha}N^{1+\kappa}-CN^{1+\kappa}N^{\frac{3\kappa-2}{2}+\frac{2-3\kappa}{4}},
    \end{equation}
    where we used $\a=\frac{\a_{0}}{N^{1-\kappa}}$ and $S(\cdot)=\sqrt{\cdot}$. Suppose $\Psi_{N}\in L^{2}_{s}(\Lambda^{N})$ is a normalized state such that 
	\begin{equation}
		\left\langle \Psi_{N}, H_{N,\kappa}\Psi_{N} \right\rangle\leq 4\pi\a_{0}N^{1+\kappa}+C_{0}N^{\frac{5\kappa}{2}+\frac{2-3\kappa}{4}}.
	\end{equation}
    Recall $\ell=K^{-1}\frac{1}{\sqrt{N\a}}\sim N^{-\kappa/2}$, we have 
    \begin{equation}\label{esitimate}
    \begin{aligned}
        \frac{\left\langle \Psi_{N},\cN_{+}\Psi_{N} \right\rangle }{N}&\lesssim \ell^{2+\alpha}N^{\kappa}+\ell^{-2-\alpha}N^{\frac{5\kappa-2}{2}+\frac{2-3\kappa}{4}}\\
        &\lesssim N^{-\frac{\alpha\kappa}{2}}+N^{(2+\alpha)\frac{\kappa}{2}+\frac{5\kappa-2}{2}+\frac{2-3\kappa}{4}}.
    \end{aligned}
    \end{equation}
    For $\kappa\in (0,\frac{2}{11})$, one can choose $\alpha>0$ to be small enough such that the right-hand side of \eqref{esitimate} converges to zero as $N$ tends to infinity. This concludes the proof.
\end{proof}

\appendix

\section{The GP Regime with Small Scattering Length}\label{sect:GP}

For the purpose of being self-contained, we include in this appendix a proof of BEC in the GP regime ($\kappa=0$)  for potential $V$ with small scattering length. This together with the proof of Theorem~\ref{thm:BEC} gives a proof of BEC beyond GP for $\kappa \in (0, 2/11)$. 

Let us consider the rescaled Hamiltonian
\begin{equation}
    H_{n,\ell}^{\rm Neu}=\sum_{i=1}^{n}-\Delta_{x_{i}}^{\rm Neu}+\sum_{1\leq  i<j\leq n}\ell^{2}V(\ell(x_{i}-x_{j}))
\end{equation}
acting on $L^{2}_{s}(\Lambda^{n})$, where $\Delta^{\Neu}=\Delta_{\Lambda}^{\Neu}$ is the Laplacian with the Neumann boundary condition. The physical space is $\Lambda$ throughout this section, and for simplicity we omit the subscript $\Lambda$.

We will prove the following result, which readily implies complete BEC when $\kappa=0$.
\begin{proposition}\label{prop:Neumann completing the square}  Let $V$ be as in Theorem \ref{thm:BEC}, whose scattering length $\a_{0}$ satisfies $\frac{\a_{0}n}{\ell}<\frac{1}{24}$. Then we have the operator lower bound 
    \begin{equation}
        H_{n,\ell}^{\rm Neu}\geq 4\pi\a_{0}\frac{n^{2}}{\ell}+\left[\mu-C\frac{n}{\ell}\frac{1}{\log (\ell)}-16\pi\a_{0}\frac{n}{\ell}\right]\mathcal{N}_{+}-C\left(\frac{n}{\ell}\right)^{2}S(\ell)-C\left(\frac{n}{\ell}\right)
    \end{equation}
    on $L_{s}^{2}(\Lambda^{n})$ for $\ell$ large enough. Here, $C>0$ is some constant independent of $n,\ell$. $\mu$ is a positive constant that satisfies $16\pi\a_{0}\frac{n}{\ell}<\mu<\pi-8\pi\a_{0}\frac{n}{\ell}$ and $S(\ell)=\sqrt{\ell}.$
\end{proposition}

\begin{remark}
Choosing $\ell\propto n$, we obtain complete BEC for the GP regime with small potentials. We only consider the Neumann boundary condition since the periodic case is standard in the GP regime.
\end{remark}

In the following, let us denote $P=\ket{\varphi_0}\!\bra{\varphi_0}$ with $\varphi_0=1\in L^2(\Lambda)$ and $Q=\1-P$. 
The key idea of the proof of Proposition~\ref{prop:Neumann completing the square}  is based on the inequality 
\begin{equation}\label{eq17}
    (\1-P\otimes P F)V_{\ell}(\1-F P\otimes P)\geq 0\, ,
\end{equation} 
where $V_{\ell}$ and $F$ are the multiplication operators by $\ell^{2}V(\ell(x-y)) \ge 0$ and $F(x,y)$ on the two-particle space. The inequality \eqref{eq17} follows the technique in \cite{nam2022optimal}, which is also inspired by earlier ideas in \cite{brietzke2020simple} and \cite{fournais2020energy} on proof of the Lee--Huang--Yang formula. This gives a quick reduction of the full two-body interaction potential to a quadratic contribution in the spirit of Bogoliubov theory \cite{bogoliubov1947theory}.
In this procedure, the choice of the correlation function $F$ plays a central role in the analysis. Unlike the work in \cite{nam2022optimal}, which focused on the finite box with periodic boundary condition and the full domain $\R^3$, here we will deal with the Neumann boundary condition. Consequently, we will construct $F$ by using the symmetrization technique introduced in \cite{haberberger2023free}. In the last step, we will handle the quadratic contribution by adapting the analysis in  \cite{haberberger2023free}. Now let us go to the details.

\subsection{Modified scattering function and Neumann symmetrization} Note that the scattering function $\omega$ in \eqref{eq:scattering-omega} satisfies the scaling property: 
\begin{equation}
	-\Delta \omega_{\ell}=\tfrac{1}{2}V_{\ell}(1-\omega_{\ell})
\end{equation}
with $\omega_{\ell}(x):=\omega(\ell x)$ and $V_{\ell}(x):=\ell^{2}V(\ell x)$. We also introduce the modified scattering solution
\begin{equation}
    \omega_{\ell,\lambda}(x)=\omega_{\ell}(x)\chi_{\lambda}(x),
\end{equation}
where $\chi_{\lambda}(x)=\chi(\lambda^{-1}x)$ with $\chi$ a fixed $C^{\infty}$ radial function satisfying
$$\chi(x)=0\ {\rm for}\ |x|\geq1\ \ {\rm and}\ \ \chi(x)=1\ {\rm for}\ |x|\leq\tfrac{1}{2}\, .$$

When $\ell$ is large enough, we have $2R/\ell<\lambda$. Using the exact formula of $\omega$, we find 
\begin{equation}\label{pointwise esitimate of omega}
	0\leq\omega_{\ell,\lambda}(x)\leq\frac{C\1_{\{|x|\leq\lambda\}}}{|\ell x|+1},
\end{equation}
where $C$ is independent of $\lambda$.
The function $\omega_{\ell,\lambda}$ satisfies the following
\begin{equation}\label{eq10}
    -\Delta \omega_{\ell,\lambda}=\tfrac{1}{2}V_{\ell}(1-\omega_{\ell})-\tfrac{1}{2}\epsilon_{\ell,\lambda},\ \ \ \ \ \tfrac{1}{2}\epsilon_{\ell,\lambda}(x)=\frac{\a_{0}}{\ell}\lambda^{-3}\left(\frac{\chi''}{|\cdot|}\right)(\lambda^{-1}x).
\end{equation}
We will choose $\lambda\leq 1$ be a constant. 

The naive choice $F=\1-\omega_{\ell,\lambda}$ in \eqref{eq17} does not work because of the Neumann boundary condition. To fix this issue, we use the Neumann symmetrization technique from \cite{haberberger2023free}. In this context, for $p \in\pi\N_{0}^{3}$, let us denote 
\begin{equation}
	\phi_{p}(x)=\prod_{i=1}^{3}\phi_{p_{i}}(x_{i}),\ \ \phi_{p_{i}}(x_{i})=\begin{cases}
	\ \ \ \ \ \ \ \ \ \ \ 1,\ & p_{i}=0\\
	\sqrt{2}\cos(p_{i}(x_{i}+1/2)),\ & p_{i}\neq 0	
	\end{cases}.
\end{equation}
The family $\{\phi_{p}\}_{p\in\pi\N_{0}^{3}}$ is an orthonormal basis for $L^{2}(\Lambda)$ satisfying the Neumann boundary conditions. Denoting 
$$\Lambda+z=\{x+z:x\in\Lambda\},\quad z\in\mathbb{Z}^{3}\, ,$$
we define the transformation 
\begin{equation}
P_{z}: \lambda\rightarrow\lambda+z,\ \ (P_{z}(x))_{i}=(-1)^{z_{i}}x_{i}+z_{i},
\end{equation}
which maps a point $x\in\Lambda$ to its mirror point in the box $\Lambda+z$.

Let $f:\mathbb{R}^{3}\rightarrow\mathbb{R}$ be radial and integrable with $\supp(f)\subset \Lambda$. Then for all $p,q\in\pi\mathbb{N}_{0}^{3}$ we have the following useful identity:
\begin{equation}\label{symidentity}
\int_{\Lambda^{2}}\sum_{z\in\mathbb{Z}^{3}}f(P_{z}(x)-y)\varphi_{p}(x)\varphi_{q}(y)\dd x\d y=\delta_{p,q}\widehat{f}(p).
\end{equation}
The proof can be found in \cite[Lemma 3.2]{haberberger2023free}. Here, we define the Fourier transform $f$ by $\hat{f}(p)=\int_{\R^{3}}f(x)e^{-ix\cdot p}\dd x.$

 We define the function $\widetilde{W}:\Lambda^{2}\rightarrow \mathbb{R} $ as 
\begin{equation}\label{defW}
    \widetilde{W}(x,y)=\sum_{z\in\mathbb{Z}^{3}}\omega_{\ell,\lambda}(P_{z}(x)-y).
\end{equation}
The function $\widetilde{W}$ is well-defined due to the finiteness of the sum. We find that 
\begin{equation}
    \widetilde{W}(x,y)=\omega_{\ell,\lambda}(x-y),\quad \forall x,y\in\{z\in\Lambda:{\rm dist}(z,\partial\Lambda)>\lambda\}\, .
\end{equation}
In fact, $\widetilde{W}(x,y)$ is diagonal in the Neumann basis $\{\varphi_{p}\}_{p\in\pi\mathbb{N}_{0}^{3}}$, that is,  
\begin{equation}
    \widetilde{W}(x,y)=\sum_{p\in\pi\mathbb{N}_{0}^{3}}\widehat{\omega}_{\ell,\lambda}(p)\varphi_{p}(x)\varphi_{p}(y)\, .
\end{equation}
Moreover, we remove from the function $\widetilde{W}$ the contribution from the zero mode via the projection $Q=\1-\ket{\varphi_{0}}\!\bra{\varphi_{0}}$ as follows 
\begin{equation}\label{eq16}
    W(x,y)=(Q^{\otimes 2}\widetilde{W})(x,y)=\sum_{p\in\pi\mathbb{N}_{0}^{3}\setminus\{0\} }\widehat{\omega}_{\ell,\lambda}(p)\varphi_{p}(x)\varphi_{p}(y)\,.
\end{equation}

Note that $|P_{z}(x)-y\vert\geq|x-y\vert$ for all $x,y\in\Lambda$. Together with \eqref{pointwise esitimate of omega} and the finiteness of the sum yields
\begin{equation}\label{eq34}
    |\widetilde{W}(x,y)\vert\leq\frac{C\1_{|x-y|\leq\lambda}}{1+\ell|x-y|}.
\end{equation}
Moreover, \eqref{pointwise esitimate of omega} yields 
\begin{equation}\label{eq35}
    \widetilde{W}(x,y)-W(x,y)=\widehat{\omega}_{\ell,\lambda}(0)\leq C\frac{\lambda^{2}}{\ell}.
\end{equation}

Finally, we take $F(x,y)=1-W(x,y)$, which is the desired modified correlation function. We can recover the scattering length by $F$ as follows:

\begin{lemma}[Boundary effects]\label{l2}
    For the function 
    \begin{equation}
        h(x)=\int_{\Lambda}nV_{\ell}(x-y)F(x,y)\dd y-8\pi\a_{0}\frac{n}{\ell}
    \end{equation}
    defined on $\Lambda$, {there exists $C_p>0$, independent of $n$ and $\ell$,} such that the following holds
    $$\|h\Vert_{1}\leq C\frac{n}{\ell}\frac{\log(\ell)}{\ell} \quad \text{ and } \quad \|h\Vert_{p}\leq C\frac{n}{\ell}\ell^{-1/p},\ \forall p\in(1,\infty]\,.$$
    Consequently, we have that
    $$\left|\int_{\Lambda^{2}}nV_{\ell}(x-y)F(x,y)\dd x\d y-8\pi\a_{0}\frac{n}{\ell}\right|\leq C\frac{n}{\ell}\frac{\log (\ell)}{\ell}\,.$$ 
\end{lemma}
\begin{proof}
    The proof is exactly the same as in \cite[Lemma 3.3]{haberberger2023free}, with $n+K(x,y)=nF(x,y)$ in our notations.
\end{proof} 

\subsection{Proof of Proposition~\ref{prop:Neumann completing the square} }
Let 
$a^{\ast}_p$ and $a_p$
be the usual creation and annihilation operators acting on the bosonic Fock space $\cF_s :=\C\oplus\bigoplus^\infty_{n=1}L^2_s(\Lambda^{n})$ associated with the basis vector $\varphi_p$. They satisfy the canonical commutation relations: $[a_p, a^\ast_q] = \delta_{p, q}, [a_p, a_q] = [a^\ast_p, a^\ast_q] =0$, for all $p, q \in \pi \N_0^3$. The particle number and the excitation number operators are defined, respectively, by 
\begin{equation}
	\cN=\sum_{p\in\pi\N_{0}^{3}}a_{p}^{*}a_{p}\quad \text{ and } \quad \cN_{+}=\sum_{\substack{p \in\pi\N_{0}^{3}\\ p\neq 0}}a_{p}^{*}a_{p}\,.
\end{equation}
Using these operators, we can rewrite $H_{n,\ell}^{\rm Neu}$ as a restriction of an operator on Fock space to the $n$-particle space. Expanding \eqref{eq17}) leads to 

\begin{align*}
    H_{n,\ell}^{\rm Neu}\geq&\, \left\{\left(\frac{1}{2}\int_{\Lambda^{2}}\left(2F(x,y)-F(x,y)^{2}\right)V_{\ell}(x-y)\dd x\d y\right) a_{0}^{*}a_{0}^{*}a_{0}a_{0}\right\}\\
      &+\left\{a^{*}\left[Q\left(\int_{\Lambda}F(\cdot,y)V_{\ell}(\cdot-y)\dd y\right)\right]a_{0}^{*}a_{0}a_{0}+{\rm h.c.}\right\}\\
      &+\left\{\sum_{p}|p|^{2}a_{p}^{*}a_{p}+\frac{1}{2}\sum_{p,q\neq0}\left(\left\langle \varphi_{p}\otimes\varphi_{q},\widetilde{K}(x,y) \right\rangle\frac{1}{n}a_{p}^{*}a_{q}^{*}a_{0}a_{0}+{\rm h.c.} \right)\right\}\\
      :=&\, \mathcal{H}_{0}+\mathcal{H}_{1}+\mathcal{H}_{2},
\end{align*}
where $\widetilde{K}(x,y)=nV_{\ell}(x-y)F(x,y)$. Firstly, we focus on $\mathcal{H}_{0}$ and $\mathcal{H}_{1}$.

\begin{lemma}\label{l3}
    Suppose $\lambda$ is small enough, we have 
    \begin{multline}
                \cH_{0}+\cH_{1}\geq\frac{n^{2}}{2}\int_{\Lambda^{2}}\left(2F-F^{2}\right)V_{\ell}\dd x \d y-\left(16\pi\a_{0}\frac{n}{\ell}+\frac{n}{\ell}\frac{C}{\log(\ell)}\right)\mathcal{N}_{+}\\
                -C\left(\frac{n}{\ell}\right)^{2}\log(\ell)-C\left(\frac{n}{\ell}\right).
    \end{multline}
\end{lemma}

\begin{proof}
    Since $2R/\ell<\lambda$, we find that $\widetilde{W}(x,y)=\omega_{\ell,\lambda}(x-y)$ whenever ${\rm dist}(x,\partial \Lambda)>2\lambda$. Using the finiteness of the sum in \eqref{defW}, we obtain a universal constant $C$ such that
    \begin{align*}
    \int_{\Lambda^{2}}\left(2F-F^{2}\right)V_{\ell}&=\int_{\Lambda^{2}}\left(1-W^{2}\right)V_{\ell}\\
    &\geq\frac{1}{2\ell}\int_{\R^{3}}(2f-f^{2})V-\frac{C\lambda}{\ell}\int_{\R^{3}}(1+\omega+\omega^{2})V.
    \end{align*}
    Choosing $\lambda$ to be small enough, we conclude that
    $$\int_{\Lambda^{2}}\left(2F-F^{2}\right)V_{\ell}=\int_{\Lambda^{2}}\left(1-W^{2}\right)V_{\ell}\geq0.$$
    We have 
    \begin{equation}\label{eq32}
    \begin{aligned}
        \mathcal{H}_{0}=&\, \left(\frac{1}{2}\int_{\Lambda^{2}}\left(2F(x,y)-F(x,y)^{2}\right)V_{\ell}(x-y)\dd x\d y\right) \left(n-\mathcal{N}_{+}\right)(n-\mathcal{N}_{+}-1)\\
        \geq&\, \frac{n^{2}-n}{2}\int_{\Lambda^{2}}\left(2F(x,y)-F(x,y)^{2}\right)V_{\ell}(x-y)\dd x\d y\\
        &\, -2n\left(\int_{\Lambda^{2}}F(x,y)V_{\ell}(x-y)\dd x\d y\right)\mathcal{N}_{+}\\
        \geq&\, \frac{n^{2}}{2}\int_{\Lambda^{2}}\left(2F(x,y)-F(x,y)^{2}\right)V_{\ell}(x-y)\dd x\d y\\
        &\, -\left(16\pi\a_{0}\frac{n}{\ell}+C\frac{n}{\ell}\frac{\log(\ell)}{\ell}\right)\mathcal{N}_{+}-C\left(\frac{n}{\ell}\right)\,,
    \end{aligned}
    \end{equation}
    where we have used Lemma~\ref{l2} in the last inequality.
    
    Next, we consider $\mathcal{H}_{1}$. Leaving $g=nQ\left(\int_{\Lambda}F(\cdot,y)V_{\ell}(\cdot-y)\dd y\right)$, we have $g=h-\int_{\Lambda}h$. Using Lemma~\ref{l2} again, we have 
    \begin{equation}
        \|g\Vert_{2}\leq\|h\Vert_{2}+\|g-h\Vert_{2}\leq\|h\Vert_{2}+\|h\Vert_{1}\leq C\frac{n}{\ell}\frac{1}{\sqrt{\ell}}.
    \end{equation}
    Therefore, by the Cauchy--Schwarz inequality 
    \begin{align*}
        \pm \left(a^{*}(g)a_{0}+{\rm h.c.}\right)&\leq \varepsilon a^{*}(g)a(g)+\frac{1}{\varepsilon}a_{0}^{*}a_{0}\leq\varepsilon\|g\Vert_{2}^{2}\mathcal{N}_{+}+\frac{n}{\varepsilon}\\
        &\leq\varepsilon\left(\frac{n}{\ell}\right)^{2}\frac{C}{\ell}\mathcal{N}_{+}+\frac{n}{\varepsilon}=\frac{n}{\ell}\frac{1}{\log(\ell)}\mathcal{N}_{+}+C\left(\frac{n}{\ell}\right)^{2}\log(\ell),
    \end{align*}
    where we chose $\varepsilon=\frac{\ell^{2}}{Cn\log(\ell)}.$ Similarly, 
    \begin{align*}
        \pm\left(a^{*}(g/n)\mathcal{N}_{+}a_{0}+{\rm h.c.}\right)&\leq\frac{\varepsilon}{n^{2}}\|g\Vert_{2}^{2}\mathcal{N}_{+}+\frac{n^{2}}{\varepsilon}\mathcal{N}_{+}\\
        &\leq\left(\frac{C\varepsilon}{\ell^{3}}+\frac{n^{2}}{\varepsilon}\right)\mathcal{N}_{+}=C\frac{n}{\ell}\frac{1}{\sqrt{\ell}}\mathcal{N}_{+}.
    \end{align*}
    Therefore,
    \begin{equation}
        \mathcal{H}_{1}=\left(a^{*}(g)a_{0}-a^{*}(g/n)\mathcal{N}_{+}a_{0}+{\rm h.c.}\right)\geq-\frac{n}{\ell}\frac{C}{\log(\ell)}\mathcal{N}_{+}-C\left(\frac{n}{\ell}\right)^{2}\log(\ell)
    \end{equation}
    which combined with \eqref{eq32} gives the desired bound.
\end{proof}

To deal with $\mathcal{H}_{2}$, first recall the definition of $W$ in \eqref{eq16} and the modified scattering equation \eqref{eq10}, we can rewrite $\widetilde{K}$ as follows:
\begin{equation}\label{eq33}
    \begin{aligned}
        \widetilde{K}(x,y)=&\, n\,V_{\ell}\left( 1-W \right) \\
        =&\, \sum_{z\in\mathbb{Z}^{3}}n\,(-2\Delta)\omega_{\ell,\lambda}(P_{z}(x)-y)\\
        &+\sum_{z\in\mathbb{Z}^{3}}n\,\epsilon_{\ell,\lambda}(P_{z}(x)-y)\\
        &+n\sum_{z\neq 0}\left[-V_{\ell}(1-\omega_{\ell})(P_{z}(x)-y)-V_{\ell}(x-y)\omega_{\ell,\lambda}(P_{z}(x)-y)\right]\\
        &+n\,\widehat{\omega}_{\ell,\lambda}(0)V_{\ell}(x-y)\\
        :=&\, K_{\rm m}(x,y)+2\,\widetilde{Q}_{2}^{\epsilon}(x,y)
        +2\,\widetilde{Q}_{2}^{{\rm bc}}(x,y)+n\,\widehat{\omega}_{\ell,\lambda}(0)V_{\ell}(x-y).
    \end{aligned}
\end{equation}

Using the identity \eqref{symidentity} and the modified scattering equation \eqref{eq10}, we see that 

\begin{equation}\label{eq1*}
\begin{aligned}
K_{\rm m}(x,y)+2\title{Q}_{2}^{\epsilon}(x,y)&=\sum_{z\in\Z^{3}}n\left(-2\Delta\omega_{\ell,\lambda}+\epsilon_{\ell,\lambda}\right)\left(P_{z}(x)-y\right)\\
&=\sum_{z\in\Z^{3}}n\left(V_{\ell}(1-\omega_{\ell,\lambda})\right)\left(P_{z}(x)-y\right)\\
&=\sum_{z\in\Z^{3}}n\left(V_{\ell}f_{\ell}\right)\left(P_{z}(x)-y\right)=\sum_{p}n\widehat{V_{\ell}f_{\ell}}(p)\phi_{p}(x)\phi_{p}(y).
\end{aligned}
\end{equation}
Inserting \eqref{eq33} and \eqref{eq1*} in the definition of $\cH_{2}$, we conclude that
\begin{align*}
    \mathcal{H}_{2}-\mu\mathcal{N}_{+}=&\, \left\{\sum_{p,q\neq 0}\left(\left\langle\varphi_{p}\otimes\varphi_{q}\left| \widetilde{Q}_{2}^{{\rm bc}}+\frac{n}{2}\widehat{\omega}_{\ell,\lambda}(0)V_{\ell}\right.\right\rangle\frac{1}{n}a_{p}^{*}a_{q}^{*}a_{0}a_{0}+{\rm h.c.}\right)\right\}\\
    &+\left\{\sum_{p\neq 0}(|p|^{2}-\mu)a_{p}^{*}a_{p}+\frac{1}{2}\sum_{p\neq 0}n\widehat{V_{\ell}f_{\ell}}(p)\left(\frac{1}{n}a_{p}^{*}a_{p}^{*}a_{0}a_{0}+{\rm h.c.}\right)\right\}\\
    :=&\, \mathcal{V}+\mathbb{H}.
\end{align*}
We remark here that $\mathcal{V}$ contains the effect of the Neumann boundary condition. $\mathbb{H}$ is the Bogoliubov quadratic Hamiltonian which can be diagonalized directly. To be precise, we have the following lemma:

\begin{lemma}[Quadratic diagonalization]\label{l4}
    Suppose $16\pi\a_{0}\frac{n}{\ell}<\mu<\pi-8\pi\a_{0}\frac{n}{\ell}$. Then
    $$\mathbb{H}\geq-\frac{n^{2}}{2}\int_{\Lambda^{2}}V_{\ell}(F-F^{2})-C \left( \frac{n}{\ell} \right)^{2}\log(\ell) $$
    with $C$ depending only on $\lambda$ and $\a_{0}$.
\end{lemma}

\begin{proof}
    Recall the simple case of Bogoliubov's method \cite[Theorem 6.3]{lieb2001ground}, we have 
    \begin{equation}
        \mathcal{A}\left( a_{+}^{*}a_{+}+a_{-}^{*}a_{-} \right)+\mathcal{B}\left( a_{+}^{*}a_{-}^{*}+a_{+}a_{-} \right)\geq -\left( \mathcal{A}-\sqrt{\mathcal{A}^{2}-\mathcal{B}^{2}} \right)\frac{[a_{+},a_{+}^{*}]+[a_{-},a_{-}^{*}]}{2}  
    \end{equation}
    for all operators $a_{+},a_{-}$ on Fock space and real constants $\mathcal{A}, \mathcal{B}$ such that $[a_{+},a_{-}]=0$ and $\mathcal{A}>|\mathcal{B}|$. Taking 
    $$a_{+}=a_{-}=b_{p}=n^{-1/2}a_{0}^{*}a_{p},\ \ b_{p}^{*}b_{p}\leq a_{p}^{*}a_{p},\ \ [b_{p},b_{p}^{*}]\leq 1 \ \ {\rm for\ all}\ \  0\neq p \in \pi\mathbb{N}_{0}^{3}\,,$$
    we find that 
    \begin{equation}\label{eq36}
        \begin{aligned}
            \mathbb{H}&\geq \frac{1}{2}\sum_{p\neq 0}\left( |p|^{2}-\mu \right)\left( b_{p}^{*}b_{p}+b_{p}^{*}b_{p} \right)+n\widehat{V_{\ell}f_{\ell}}(p)\left( b_{p}^{*}b_{p}^{*}+b_{p}b_{p} \right)   \\
            &\geq-\frac{1}{2}\sum_{p\neq 0}\left( |p|^{2}-\mu-\sqrt{\left( |p|^{2}-\mu \right)^{2}-|n\widehat{V_{\ell}f_{\ell}}(p)\vert^{2} } \right) \\
            &\geq-\frac{n^{2}}{4}\sum_{p\neq 0}\frac{|\widehat{V_{\ell}f_{\ell}}(p)\vert^{2}}{|p|^{2}}-C \left( \frac{n}{\ell} \right)^{2}\,.
        \end{aligned} 
    \end{equation}
    Here, we used $16\pi\a_{0}\frac{n}{\ell}<\mu<\pi-8\pi\a_{0}\frac{n}{\ell}$ and $n\|\widehat{V_{\ell}f_{\ell}}\Vert_{L^{\infty}}\leq 8\pi\a_{0}\frac{n}{\ell}$. The last inequality comes from Taylor's expansion and the constant $C$ here depends only on $\a_{0}$.

    Recalling from \eqref{eq33}, we denote by $\widetilde{K}$ the operator on $L^{2}(\Lambda)$ with the integral kernel $\widetilde{K}(x,y)$. Splitting $\widetilde{K}$ into three operators, we have $\widetilde{K}=K_{{\rm m}}+K_{\epsilon}+\widetilde{K}_{r}$ with the integral kernel being 
\begin{align*}
    &K_{m}(x,y),\ K_{\epsilon}(x,y)=2\widetilde{Q}_{2}^{\epsilon}(x,y),\\ 
    &\widetilde{K}_{r}(x,y)=2\widetilde{Q}_{2}^{{\rm bc}}(x,y)+n\widehat{\omega}_{\ell,\lambda}(0)V_{\ell}(x-y).
\end{align*}
Note that $K_{{\rm m}}$ and $K_{\epsilon}$ are invariant on $\{\varphi_{0}\}^{\bot }$. Using the same calculation as in \eqref{eq1*}, we can rewrite the right hand side of \eqref{eq36} to obtain
\begin{equation}\label{eq36*}
    \mathbb{H}\geq-\frac{1}{4}\sum_{p\neq 0}\frac{\left\langle (K_{{\rm m}}+K_{\epsilon})\varphi_{p}\left|(K_{{\rm m}}+K_{\epsilon})\varphi_{p}\right.\right\rangle}{|p|^{2}}-C \left( \frac{n}{\ell} \right)^{2}
\end{equation}

    Next, we split the infinite sum into four parts as follows.
\begin{equation}\label{eq37}
    \begin{aligned}
        &\sum_{p\neq 0}\frac{\left\langle (K_{{\rm m}}+K_{\epsilon})\varphi_{p}\left|(K_{{\rm m}}+K_{\epsilon})\varphi_{p}\right.\right\rangle}{|p|^{2}}\\
        &=\sum_{p\neq 0}\frac{\left\langle \widetilde{K}\varphi_{p}\ \left|K_{{\rm m}}\varphi_{p}\right.\right\rangle}{|p|^{2}}+\sum_{p\neq 0}\frac{\left\langle K_{{\rm m}}\varphi_{p}\left|K_{\epsilon}\varphi_{p}\right.\right\rangle}{|p|^{2}}+\sum_{p\neq 0}\frac{\left\langle K_{\epsilon}\varphi_{p}\left|K_{\epsilon}\varphi_{p}\right.\right\rangle}{|p|^{2}}-\sum_{p\neq 0}\frac{\left\langle \widetilde{K}_{r}\varphi_{p}\ \left|K_{{\rm m}}\varphi_{p}\right.\right\rangle}{|p|^{2}}\\
        &:=(\mathrm{I})+(\mathrm{II})+(\mathrm{III})+(\mathrm{IV}).
    \end{aligned}
\end{equation}
We can recover $F$ from $(\mathrm{I})$. Recall $K_{{\rm m}}$ is diagonalized in Neumann basis, whose integral kernel satisfies
\begin{equation}
    K_{{\rm m}}(x,y)=\sum_{p\neq 0}2n|p|^{2}\widehat{\omega}_{\ell,\lambda}(p)\varphi_{p}(x)\varphi_{p}(y).
\end{equation}
Therefore, $|p|^{2}$ in the denominator is canceled out and we have
\begin{equation}\label{eq38}
    \begin{aligned}
        (\mathrm{I})=\int_{\Lambda^{2}}\tilde{K}(x,y)\sum_{p\neq 0}2n\widehat{\omega}_{\ell,\lambda}(p)\varphi_{p}(x)\varphi_{p}(y)\dd x\d y=2n^{2}\int_{\Lambda^{2}}V_{\ell}(F-F^{2}).\\
    \end{aligned} 
\end{equation}
Moreover, we use the pointwise estimates in \eqref{eq34} \eqref{eq35} and the definition of $\epsilon_{\ell,\lambda}$ in \eqref{eq10} to arrive at
\begin{equation}\label{eq39}
    \begin{aligned}
        |(\mathrm{II})|&=\left|\int_{\Lambda^{2}}4nW(x,y)\widetilde{Q}_{2}^{\epsilon}(x,y)\dd x\d y\right|\leq C  \|nW\Vert_{L^{2}}\|\tilde{Q}_{2}^{\epsilon}\Vert_{L^{2}} \\
        &\leq \frac{Cn}{\lambda^{3}\ell}\sqrt{Cn^{2}\int_{\Lambda^{2}}\frac{\1_{|x-y|\leq\lambda}\dd x\d y}{\left( 1+\ell|x-y| \right)^{2} }+C\lambda^{4}\left( \frac{n}{\ell} \right)^{2} }\leq\frac{C}{\lambda^{5/2}}\left( \frac{n}{\ell} \right)^{2}\\
    \end{aligned} 
\end{equation}
and 
\begin{equation}\label{eq40}
    |(\mathrm{III})|\leq C\int_{\Lambda^{2}}|\widetilde{Q}_{2}^{\epsilon}(x,y)|^{2}\dd x\d y\leq  \frac{C}{\lambda^{6}}\left( \frac{n}{\ell} \right)^{2}.
\end{equation}
For all $0\neq z\in\mathbb{Z}^{2}$ we have 
$$V_{\ell}(P_{z}(x)-y)=V_{\ell}(P_{z}(x)-y)\1_{{\rm d}(x,\partial\Lambda)\leq R\ell^{-1}}\leq\frac{CV_{\ell}(P_{z}(x)-y)}{1+\ell{\rm d}(x,\partial\Lambda)}$$
since ${\rm supp}V_{\ell}\subset B_{R\ell^{-1}}(0)$. In combination with the non-increasing assumption on $V$, we arrive at the bound 
$$|\tilde{Q}_{2}^{{\rm bc}}(x,y)|\leq Cn\sum_{z\in \mathbb{Z}^{3}}\frac{V_{\ell}(P_{z}(x)-y)}{1+\ell{\rm d}(x,\partial\Lambda)}\leq Cn\frac{V_{\ell}(x-y)}{1+\ell{\rm d}(x,\partial\Lambda)}.$$ 
Using $\|W\Vert_{L^{\infty}}\leq C$, which can be obtained from the pointwise estimate, we have
\begin{align*}
    \left|4n\int_{\Lambda^{2}}W(x,y)\widetilde{Q}_{2}^{{\rm bc}}(x,y)\dd x\d y\right|&\leq Cn^{2}\int_{\Lambda^{2}}\frac{V_{\ell}(x-y)}{1+\ell{\rm d}(x,\partial\Lambda)}\\
    &\leq Cn^{2}\ell^{-1}\int_{\Lambda}\frac{1}{1+\ell{\rm d}(x,\partial\Lambda)}\leq c \left( \frac{n}{\ell} \right)^{2}\log(\ell). \\
\end{align*} 
Therefore,
\begin{equation}\label{eq41}
    \begin{aligned}
        |(\mathrm{IV})|&=\left|4n\int_{\Lambda^{2}}W(x,y)\widetilde{Q}_{2}^{{\rm bc}}(x,y)\dd x\d y+2n\int_{\Lambda^{2}}W(x,y)n\widehat{\omega}_{\ell,\lambda}(0)V_{\ell}(x-y)\dd x\d y\right|\\
        &\leq C\left( \frac{n}{\ell} \right)^{2} \log(\ell)+\frac{Cn^{2}}{\ell}\int_{\Lambda^{2}}V_{\ell}(x-y)\dd x\d y\leq C\left( \frac{n}{\ell} \right)^{2} \log(\ell),
    \end{aligned}
\end{equation}
where we used $\widehat{\omega}_{\ell,\lambda}(0)\leq C\lambda^{2}/\ell$.

Inserting \eqref{eq37} and \eqref{eq38}--\eqref{eq41} into \eqref{eq36*}, we obtain the desired result.
\end{proof}

Finally, again using the pointwise estimate
$$\left|\tilde{Q}_{2}^{\rm bc}+\frac{n}{2}\widehat{\omega}_{\ell,\lambda}(0)V_{\ell}\right|\leq\frac{CnV_{\ell}(x-y)}{1+\ell{\rm d}(x,\partial\Lambda)},$$
we find by the Cauchy--Schwarz inequality
$$\pm\mathcal{V}\leq \delta H_{n,\ell}^{{\rm Neu}}+\frac{C}{\delta}\left( \frac{n}{\ell} \right)^{2} $$
for any $\delta>0$. This, together with Lemma~\ref{l2}, Lemma~\ref{l3} and Lemma~\ref{l4} concludes the proof of Theorem~\ref{prop:Neumann completing the square}, after optimizing over $\delta$.

 \section{Discrete Poincaré inequality}\label{appendix:discretepoincare}
 Given $M,d\in\N$, we define a graph $\Gamma=\llbracket M\rrbracket^d$, namely 
 \begin{equation}
     \Gamma:=\left\{1,2,\cdots, M\right\}^{d},\ \ x\sim y\ {\rm iff }\ \exists \ i \ {\rm s.t.}\abs{x_{j}-y_{j}}=\delta_{ij},
 \end{equation}
 where $x,y\in\Gamma$ are two arbitrary points and $\delta_{ij}$ is the Kronecker delta function. Suppose $f:\Gamma\rightarrow\C$ is a complex-valued function on $\Gamma$, we define its mean $f_{\rm Avg}$ by
 \begin{equation}
     f_{\rm Avg}:=\frac{1}{\abs{\Gamma}}\sum_{x\in\Gamma}f(x)
 \end{equation}
 with $\abs{\Gamma}=M^{d}$ being the cardinality of the set $\Gamma$.

\begin{proposition}\label{prop:discretepoincare}
     Let $p\in [1,\infty)$. There exist a constant $C=C(p,d)$, depending only on $p$ and $d$, such that 
     \begin{equation}
         \sum_{x\in\Gamma}\abs{f(x)-f_{\rm Avg}}^{p}\leq C M^{p}\sum_{x\in\Gamma}\left( \sum_{y:y\sim x}\abs{f(x)-f(y)} \right)^{p} 
     \end{equation}
     for all $\C$-valued function $f$ on $\Gamma$.
 \end{proposition}

 \begin{proof}
     For $s,x\in\Gamma$, we have the following identity:
     \begin{equation}
         f(x)-f(s)=\sum_{n=1}^{d}\sum_{t_{n}=s_{n}}^{x_{n}-1}\partial_{n}f \left( x_{1},\cdots,x_{n-1},t_{n},s_{n+1},\cdots,s_{d} \right),
     \end{equation}
     where 
     \begin{equation}
         \begin{aligned}
             \partial_{n}f \left( x_{1},\cdots,x_{n-1},t_{n},s_{n+1},\cdots,s_{d} \right)&:=f \left( x_{1},\cdots,x_{n-1},t_{n},s_{n+1},\cdots,s_{d} \right)\\
             &-f \left( x_{1},\cdots,x_{n-1},t_{n}+1,s_{n+1},\cdots,s_{d} \right).
         \end{aligned}
     \end{equation}
     We obtain by H\"older's inequality that 
     \begin{equation}
         \abs{f(x)-f(s)}^{p}\leq d^{\frac{p}{q}}\sum_{n=1}^{d}\left(\sum_{t_{n}=1}^{M}\abs{\partial_{n}f}\right)^{p}\leq(dM)^{\frac{p}{q}}\sum_{n=1}^{d}\sum_{t_{n}=1}^{M}\abs{\partial_{n}f}^{p}.
     \end{equation}
     Here $p,q$ are the H\"older conjugates; $p,q\in(0,\infty)$ such that $\frac{1}{p}+\frac{1}{q}=1$. Summing over $x$ and $s$, we have 
     \begin{equation}
         \sum_{x,s\in\Gamma}\abs{f(x)-f(s)}^{p}\leq(dM)^{\frac{p}{q}}\sum_{n=1}^{d}\sum_{x,s\in\Gamma}\sum_{t_{n}=1}^{M}\abs{\partial_{n}f}^{p}. 
     \end{equation}
     For fixed $n$, the functon $\partial_{n}f$ is independent of $\{s_{1},\cdots,s_{i},x_{i},\cdots,x_{d}\}$. A direct calculation implies 
     \begin{equation}
         \begin{aligned}
         \sum_{x,s\in\Gamma}\sum_{t_{n}=1}^{M}\abs{\partial_{n}f}^{p}&=\sum_{x_{1}=1}^{M}\cdots\sum_{x_{d}=1}^{M}\sum_{s_{1}=1}^{M}\cdots\sum_{s_{d}=1}^{M}\sum_{t_{n}=1}^{M}\abs{\partial_{n}f}^{p}\\
         &=M^{d+1}\sum_{x_{1}=1}^{M}\cdots\sum_{x_{n-1}=1}^{M}\sum_{t_{n}=1}^{M}\sum_{s_{n+1}=1}^{M}\cdots\sum_{s_{d}=1}^{M}\abs{\partial_{n}f}^{p}\\
         &\lesssim_{d}M\abs{\Gamma}\sum_{x\in\Gamma}\left( \sum_{y:y\sim x}\abs{f(y)-f(x)} \right)^{p}. 
         \end{aligned}
     \end{equation}
     Thus 
     \begin{equation}
         \begin{aligned}
             \sum_{x\in\Gamma}\abs{f(x)-f_{\rm Avg}}^{p}&\leq\frac{1}{|\Gamma|}\sum_{x,s\in\Gamma}\abs{f(x)-f(s)}^{p}\\
             &\lesssim_{d,p}M^{\frac{p}{q}+1}\sum_{x\in\Gamma}\left( \sum_{y:y\sim x}\abs{f(y)-f(x)} \right)^{p}\\
             &\lesssim_{d,p} M^{p}\sum_{x\in\Gamma}\left( \sum_{y:y\sim x}\abs{f(y)-f(x)} \right)^{p}.
         \end{aligned}
     \end{equation}
     Here we used H\"older's inequality again and $\frac{p}{q}+1=p$. This concludes the proof of Proposition \ref{prop:discretepoincare}.
 \end{proof}



\end{document}